\documentclass[11pt]{article}

\usepackage{amsmath, color, cite}
\usepackage{amsfonts}
\usepackage{amsthm}
\usepackage{amssymb}
\usepackage{mathrsfs}

\newtheorem{theorem}{Theorem}[section]

\newtheorem{lemma}[theorem]{Lemma}

\newtheorem{problem}[theorem]{Problem}
\numberwithin{equation}{section}

\def\Div{\mathop{\rm Div}\nolimits}

\def\bar{\overline}

\def\calA{\mathcal{A}}
\def\calB{\mathcal{B}}

\def\eps{\varepsilon}

\def\le{\leqslant}
\def\ge{\geqslant}

\begin{document}

\title{A frictional contact problem with wear diffusion}

\author{
Piotr Kalita$^{\,1}$, Pawel Szafraniec$^{\, 1}$ and 
Meir Shillor$^{\, 2}$
 \\ ~ \\
{\small $^1$ Jagiellonian University, Faculty of Mathematics and Computer Science} \\
{\small 30348 Krakow, Poland} \\
{\small $^2$ Department of Mathematics and Statistics, Oakland University} \\
{\small Rochester, MI 48309, USA} 
}

\date{\today}

\maketitle

\bibliographystyle{plain}

\begin{abstract}
  This paper constructs and analyzes a model for the dynamic frictional contact between a
  viscoelastic body and a moving foundation. The contact involves wear of the contacting 
  surface and the diffusion of the wear debris. The relationships between the stresses and 
  displacements on the contact boundary are modeled by the normal compliance law and  
  a version of the Coulomb law of dry  friction. The rate of wear of the contact surface is 
  described by the differential form of the Archard law. The effects of the diffusion of the 
  wear particles that cannot leave the contact surface on the surface are taken 
  into account. The novelty of this work is that the contact surface is a manifold and, 
  consequently, the diffusion of the debris takes place on a curved surface.
   The  interest in the model is related to the wear of mechanical joints
  and orthopedic biomechanics where the wear debris are trapped, they diffuse and often
  cause the degradation of the properties of joint prosthesis and various implants. The 
  model is in the form of a differential inclusion for the mechanical contact and the diffusion 
  equation for the wear debris on the contacting surface.
  The existence of a weak solution is proved by using a truncation argument and the 
  Kakutani--Ky Fan--Glicksberg fixed point theorem.
  \end{abstract}



\section{Introduction}
\label{Intro}

This work studies a nonlinear dynamical model for the process of contact 
between a viscoelastic body and a reactive foundation when wear debris is 
generated and diffuses on the contact surface. The model includes subdifferential 
friction boundary condition, and considerably extends the 
model and the results in  \cite{aShillorSofoneaTelega2004a}, which were 
announced in \cite{SST03} and further developed in \cite{gasinski2, GOS-wear-15}. 
Additional information and details can be found 
in \cite{SST04-book}.  The research in \cite{aShillorSofoneaTelega2004a} was 
motivated, in part, by biomechanical applications.
Indeed, such problems arise in artificial joints after arthroplasty (knee, hip, shoulder,
elbow, etc.) where debris is produced by articulating parts of the prosthesis
and is transported to the bone-implant interface. The debris causes the
deterioration of the interface, and is believed to be an important factor
leading to prosthesis loosening (see, e.g.,  \cite{RT01, RTS01}  and references therein).
Thus,  there is a considerable interest in modeling such complex
contact problems arising in  implanted joints.  This pertains to both
cement-less (the so-called "press-fit") and cemented implants.

We present a mathematical model for the dynamics of such problems. The  contact
process is assumed to include friction and wear between a viscoelastic 
body and a reactive foundation. Contact is described with a generalized compliance 
condition and friction with a general subdifferential law. We assume that the wear 
generation process takes place only on a part of the contact surface, and the wear 
rate is described by a generalized  differential Archard condition that allows for the 
diffusion of the debris on the whole of the contact surface. This is the main novelty 
in the model. Such phenomena of wear diffusion can be found in many engineering 
settings, but in mathematical publications on contact and wear, it is tacitly assumed that 
the wear debris is removed from the surface once they are formed, which is the case 
some cases, such as car engines where the oil transports the debris away. The only 
mathematical works (that we are aware of) in which the wear debris 
remains on the surface and its diffusion is taken into account are
\cite{aShillorSofoneaTelega2004a, SST03}, but there the contact surface was
assumed to be planar. However, in most cases in applications, and those in joint
replacements, the surfaces are curved. Moreover, in \cite{aShillorSofoneaTelega2004a,SST03} 
the authors considered a quasistatic process and a moving foundation. 

The novelty of this paper lies in that the process is assumed to be dynamic, the 
contact  surface is a manifold and so we use of surface gradients and the Laplace--Beltrami 
operator instead of the linear diffusion equation.  Also, we use a general nonmonotone 
subdifferential conditions to model friction, which is an extension of the classical formulation 
as a variational inequality with a subdifferential in the sense of convex analysis.  In addition,
the method of proof is new and very different from the usual one based on the use of results 
for variational inclusions.

The model for the processes consists of two coupled equations: the first is the dynamic equation
of motion of a viscoelastic body and it contains a frictional multivalued 
term. The second one describes the diffusion of the wear debris on the contact surface of the 
body. Our key result is the theorem on the existence of a weak solution to the problem. In 
contrast to \cite{gasinski2, GOS-wear-15} (where the debris diffusion is modeled, but the contact 
surface is assumed to be flat) we do not use the Banach fixed point argument, but we base our 
approach on the Kakutani--Ky Fan--Glicksberg theorem that allows us to remove of the limitations 
on the constants present in the model at the cost of getting only existence, and not the uniqueness 
of a solution. In such a way we present a new way to obtain existence results for contact problems 
with friction and wear  diffusion. 

We remark here that we do not take into account adhesion effects in the model, and in many 
contact problems, one should also take into account the process of adhesion that is coupled 
with friction and wear diffusion. For instance, clinical practice shows that adhesion plays an
important role at the bone-implant interface, and for further details we
refer to  \cite{RT01, RTS01} and the references therein.

The main mathematical difficulties of this paper lie in the formulation of the wear diffusion not on 
a subset of $\mathbb{R}^2$, but on a 2D manifold in $\mathbb{R}^3$. Also, due to the fact 
we do not impose any 
smallness condition on the constants in the model, we cannot use the Banach fixed point argument 
(such as it is done in \cite{MOSBOOK}) that also asserts the solution uniqueness. In our approach 
we do not need any assumptions on the smallness of the data, but we obtain only the 
existence of a weak solution. Due to the rather general assumptions on the nonlinearities appearing 
in the problem, we are forced to use a truncation, and we first obtain the solutions to the truncated 
problem. We then obtain the necessary a priori estimates and remove this restriction by 
passing to the limit with truncation parameter.

The paper is organized as follows.   Section \ref{Model} describes the
`classical model' for the process. We also describe shortly the equation for the
wear diffusion on the contact manifold. Section 3  lists the assumptions on the
problem data and derives its variational formulation, Problem $P_{V}$. 
It is a system coupling an evolutionary differential inclusion
for the displacements  with a diffusion equation on the curved contact surface for the wear. 
Our main result, Theorem \ref{mainth}, states that under certain reasonable assumptions 
on the setting and problem data, there exists a solution of the variational problem, which 
is a weak solution for the `classical' model. The proof to the main existence result
is done in Section 4, and is based on the approach described above. 
Finally, Section 5 concludes with a short discussion and some open problems for further study.

\section{The model}
\label{Model}

  We consider a viscoelastic body that occupies a bounded domain
  $\Omega\subseteq \mathbb{R}^{d}$, $d=2,3$ that is acted upon by volume forces and 
  surface tractions. Although the case $\mathbb{R}^{2}$ is of interest mathematically, in this case
  the contact surface is a curve and there doesn't seem to be applied interest in such a case, so
  we have $d=3$ in mind. As a result, the body may come in frictional contact with a foundation and, 
  consequently, a part of the contacting surface may undergo wear. The wear particles or debris 
  produced in this process remain on the contact surface and undergo diffusion. Thus,  grooves 
  and  surface damage occur causing changes in the shape and properties of the contacting surface.
   We construct a mathematical model for the evolution of the mechanical state of
  the body during the time interval $[0,T]$, where $0<T<+\infty$.
  The unknowns in the problem are the displacements and the surface wear function.  We refer to
  \cite{aShillorSofoneaTelega2004a} for a more thorough discussion and 
  additional details of the process. The main novelty here is that the contact surface is curved, 
  while there and in \cite{gasinski2, GOS-wear-15} the contact surface was assumed to be flat, and
  moreover, here the process is dynamic.

    We let $\Gamma$ denote the boundary of $\Omega$ that is assumed to be
  Lipschitz continuous. We assume that $\Gamma$ consists of three pairwise disjoint  sets:
  $\overline{\Gamma}_D$ where the body is held fixed and $\mu_{d-1}(\Gamma_D)>0$;  $\overline{\Gamma}_N$ 
  where surface tractions act; and $\overline{\Gamma}_C$ that is the potential contact 
  surface, where friction and wear take place. The set $\Gamma_C$ is assumed to be a $C^2$ manifold with 
  smooth  boundary $\partial \Gamma_C$.   We note here that the assumption $\mu_{d-1}(\Gamma_D)>0$ is 
  not essential, but it allows to avoid certain technical difficulties, such as the lack of the Korn inequality. We 
  use the notation $\Omega_T=\Omega\times (0,T)$,
  $\Gamma_T=\Gamma\times (0,T)$, and similarly for $\Gamma_{DT}, \Gamma_{NT}$ and $\Gamma_{CT}$.

  The body is held clamped on $\Gamma_{D}$ and so the displacement field vanishes there.
  A  volume force of density $f_0$ acts in $\Omega_T$ and surface tractions 
  of density $f_N$ are applied on $\Gamma_{NT}$.
  An initial gap function $g$ can exist between the potential contact surface $\Gamma_{C}$
  and the foundation and it is measured along the outward normal $\nu$.
 
 We denote the displacement vector by $u:\overline{\Omega}\times [0,T]\to \mathbb{R}^d$, the velocity vector by
 $v=u'$, where the prime represents the time derivative,
 the linearized strain tensor by
 \[
 \eps= \eps(u)=(\varepsilon_{ij}),\qquad \varepsilon_{ij}=\frac{1}{2}(\nabla u + \nabla u^\intercal),
 \]
 so that $\eps'(u) =\eps(v)$, and the stress tensor by $\sigma=(\sigma_{ij})$, all defined 
 on $\overline{\Omega}_T$.
  
 We write
  the normal components and tangential vectors on the boundary $\Gamma_C$ as
  \[
  u_\nu=\nu\cdot u, \quad u_\tau=u-u_\nu \nu,\quad v_\nu=\nu \cdot v, \quad v_\tau=v-v_\nu \nu,
  \]
  and the normal and tangential stresses as
  \[
  \sigma_\nu=\sigma_{ij}\nu_i\nu_j,\quad \sigma_\tau=\sigma\cdot \nu - \sigma_\nu \nu.
  \]
  
  We assume that the material is viscoelastic with linear constitutive relation
  \begin{equation}
 \sigma(t) \ = \calA(\eps({v}(t)))+\calB(\eps(u(t))),\label{stressdefin}
  \end{equation}
 where, 
 \[
  \calA=(a_{ijkl}),\qquad \calB=(b_{ijkl}),
 \]
are the viscosity and elasticity tensors, respectively.  Thus,
\[
\sigma_{ij}= a_{ijkl}\varepsilon'_{kl} + b_{ijkl}\varepsilon_{kl},
\]
and summation over repeated indices is implied.  \ The viscosity and elasticity tensors 
satisfy the following assumptions. 

\begin{itemize}{
		\item[(H1)] $a_{ijkl},b_{ijkl}\in L^\infty(\Omega)$,
		\item[(H2)] $a_{ijkl}=a_{jikl}=a_{klij},\  b_{ijkl}=b_{jikl}=b_{klij}$
		for  $i,j,k,l=1,...,d$,
		\item[(H3)] $a_{ijkl}\xi_{ij}\xi_{kl} \geq \alpha|\xi|^2,\        
		b_{ijkl}\xi_{ij}\xi_{kl} \geq 0$ for $\alpha>0$ and all symmetric matrices $(\xi_{ij})_{i,j=1}^d$.}
\end{itemize}
 
\smallskip
We remark that the presence of the viscosity tensor $\mathcal{A}$ that is coercive is crucial in the proof 
of our main result. Some problems with hyperbolic inclusions have been recently studied in e.g.,\cite{noncoercive1,noncoercive2}, yet in our case they are not applicable and it remains an open 
problem to remove the viscosity term and consider a purely elastic material.

The displacement $u$ satisfies the momentum law
\begin{equation}
u''- \Div\sigma = f_0,
\end{equation}
where $f_0:\Omega\times [0,T]\to \mathbb{R}$ describes a volume force. 
The body is clamped on $\Gamma_D$, hence,
\begin{equation}
u = 0 \quad \mbox{on} \ \Gamma_D.
\end{equation}
and the traction $f_2$ ise applied on $\Gamma_N$,
\begin{equation}
\sigma\nu = f_2 \quad \mbox{on} \ \Gamma_N.
\end{equation}
 \vskip4pt
 
  We turn to describe the wear process and note that in \cite{SST03, aShillorSofoneaTelega2004a, 
  GOS-wear-15}   the contact surface $\Gamma_{C}$ was divided into two sub-domains $D_d$ and 
  $D_w$ and the wear took place only on the part $D_w$, while the diffusion of particles took place 
  on the whole of $\Gamma_{C}$.  In this work we assume that wear is generated and diffuses on 
  $\Gamma_{C}$, however, we note that it is straightforward to restrict wear generation to a part of 
  $\Gamma_C$ by introducing the appropriate characteristic function, as was done in the articles 
  above.
  
  Before we continue, since we are interested in the diffusion of the wear debris on the surface,
  we need to introduce the concepts and notation related to diffusion on curved surfaces.
  We follow \cite{Kalita-15} (see also the references therein) and in particular,  we  refer the reader to 
  \cite[ p.388]{GT83}, for the definition of hypersurfaces in $\mathbb{R}^n$ and surface gradients 
  on them. Let $S$ be a smooth surface in $\mathbb{R}^d$, if $G$ is a smooth function defined in 
  a neighborhood of $S$, the {\it surface or tangent gradient} on $S$ is defined as
  \[
  \nabla_S G =\nabla G -G_\nu \nu
  \]
  where $G_\nu=\nu \cdot \nabla G$ is the normal derivative of $G$ on $S$, recalling that
  $\nu$ denotes the unit outer normal vector to $S$.
  Thus, the surface gradient at $x \in S$ is the projection of the gradient at $x$ onto the 
  tangent plane to $S$ at $x$. Note, that for the above definition of $ \nabla_S G$ to make sense, 
  we need to extend $G$ from $S$ to an open neighborhood in $\mathbb{R}^d$, however, such 
  an extension always exists for smooth $S$ and the value of $ \nabla_S G$ does not depend on 
  the choice of the extension (see, e.g.,  \cite{Dz88}).  If we denote the components of the {\it surface
  gradient} by
  \[
   \nabla_S G=(D_i G) \qquad i=1,\ldots, d
  \] 
then the Laplace-Beltrami operator, which describes the spatial part of diffusion on the surface,
is defined by the {\it surface divergence} of the surface gradient, i.e.,
\[
\Delta_S G=  \nabla_S\cdot  \nabla_S G=D_kD_k G,
\]
where $k=1,\ldots, d$, and summation is implied.
Next, we assume that the manifold $S$ has a smooth boundary $\Gamma_S= \partial S$ and denote 
by $\nu_S$ the unit  outer normal to $S$ on $\Gamma_S$. Then, Green's formula  on $S$
is given by (see, e.g.,  \cite{DzE07})
\[
\int_S \left(\psi \Delta_S \varphi +   \nabla_S \psi \cdot  \nabla_S\varphi \right)\, dS=
\int_{\Gamma_S} \psi \nu_S \cdot   \nabla_S \varphi \,d\Gamma,
\]
and holds for each pair $(\varphi, \psi)$ of smooth functions  defined in a neighborhood of $S$. 
In our setting, $S=\Gamma_C$ and $\Gamma_S=\partial \Gamma_C$. For the sake of 
somewhat simplified notation  we use $ \nabla_\Gamma$ for the gradient and $\Delta_\Gamma$
for the Laplace--Beltrami operator on $\Gamma_C$. We use the notion of the Sobolev space 
$H^1(S)$ of functions on the manifold $S$, i.e., functions in $L^2(S)$ that have their surface 
gradients belongs to $L^2(S)^d$, see \cite{aubin} for the definition and properties of these 
functions on manifolds without boundary, and \cite{aubin2} for manifolds with smooth boundary, 
which is the case here.

  To describe the wear process and its diffusion, we introduce the wear function $\theta$ that is
  defined on the contact surface $\Gamma_{C}\times [0, T]$, and its evolution is governed by a 
  parabolic differential equation, and a zero flux boundary condition on $\partial\Gamma_C$, 
  \begin{equation}
  \frac{\partial \theta}{\partial\nu_\gamma} = 0 \quad \mbox{on} \ \partial \Gamma_C,
  \end{equation}
  since the debris cannot leave $\Gamma_C$.  We note that the rate form of the usual Archard's 
  law of wear (see, e.g., \cite{SST04-book}) states that the rate of surface wear is proportional to the frictional 
  traction, and the relative velocity, i.e., the power of the friction resistance force, and is given by
  \[
 h_w = \eta \mu p_{\nu}(u_{\nu}-g) |{v}_{\tau}(t)|,
  \]
where $\eta$ is the wear rate constant, $\mu$ is the friction coefficient, the function $p_{\nu}$ 
describes the normal stress, and more details are below, and   ${v}_{\tau}(t)$ is the tangential velocity. 
 As was done in \cite{SST03, aShillorSofoneaTelega2004a, GOS-wear-15}, we extend the Archard law
 and allow diffusion of the wear debris on the surface, i.e., we generalize $h_w$ to a function satisfying

\smallskip
\begin{itemize}
\item[(H4)] $h_w\colon \mathbb{R}^d\times\mathbb{R}^d\times \to \mathbb{R}$ is continuous and for some $C_w>0$, and  for every $u,v\in \mathbb{R}^d, \theta\in \mathbb{R}$,
$|h_w(u,v)|\le C_w(1+|u|^2+|v|^2)$.
\end{itemize}
It is straightforward to see that when the function $p_{\nu}$ has at most linear growth, 
the wear source function satisfies this assumption. Then, the extended version of the  Archard law for a 
pointwise wear process of growth and diffusion on  $\Gamma_C$ is given by
\begin{equation}
 \theta'-  \kappa \Delta_\Gamma \theta = h_w(u,v),
 \end{equation}
where $\kappa$ is the wear diffusion constant. We note that the debris source $h_w$ 
depends on the wear and the surface speed, since the wear changes the surface 
geometry, the debris changes the friction resistance, and 
the friction coefficient is known to depend on the speed.
\medskip

We turn to the contact conditions on $\Gamma_C$.
We describe the contact process on $\Gamma_C$ by a general condition of the form 
\begin{equation}
-\sigma_\nu = h_\nu(u).
\end{equation}
We impose the following hypotheses on $h_\nu$:
\begin{itemize}
	\item[(H5)] $h_\nu\colon \mathbb{R}^d \to \mathbb{R}$ is continuous and 
	$|h_\nu(u)|\le C_\nu (1+|u|)$ for every $u,v\in \mathbb{R}^d, \theta\in \mathbb{R}$, 
	for some $C_\nu>0$.
\end{itemize}

An example of a law satisfying this conditions is the normal compliance condition 
(see, e.g., \cite{SST04-book} and the references therein),
\[
\sigma_\nu=p_\nu(u_\nu-g),
\]
where $p_\nu$ is the normal compliance function that vanishes for negative arguments, 
since then there is no contact between the body and the foundation at the point of 
$\Gamma_C$. In the literature it was typically assumed to be of the form
\[
p_\nu(u_\nu-g)=\lambda_{\nu c} (u_\nu-g)_+^m,
\]
where $(\cdot)_+$ was the positive part, $\lambda_{\nu c}$ was assumed to be a large number 
and $m\geq 1$ was the normal compliance exponent (see also\cite{SST04-book}).

We describe  friction with a general subdifferential law
\begin{equation}
-\sigma_\tau \in h_\tau(u,v,\theta)\partial j(v_\tau),
\end{equation}
where $j$ is a locally Lipschitz function and $\partial j$ stands for its Clarke subdifferential (see Section 3
below for details). We suppose that  $h_\tau$ and $j_\tau$ satisfy
\begin{itemize}
	\item[(H6)]  $h_\tau\colon \mathbb{R}^d\times \mathbb{R}^d\times \mathbb{R}\to \mathbb{R}_+$ is a continuous function and $|h_\tau(u,v,\theta)|\le C_\tau (1+|u| + |v|+|\theta|)$ for every $u,v\in \mathbb{R}^d, \theta\in \mathbb{R}$, for some $C_\tau >0$ and 
	\item[(H7)] $j_{\tau}\colon\Gamma_{C}\times \mathbb{R}^d\longrightarrow \mathbb{R}$ is a function such that $j_{\tau}(\cdot,\xi)$ is measurable on $\Gamma_{C}$ for every $\xi\in \mathbb{R}^d$, $j_{\tau}(x,\cdot)$ is locally Lipschitz on $ \mathbb{R}^d$ for a.e. $x\in\Gamma_{C}$ and moreover $\zeta \cdot \xi \ge 0$ for $\zeta \in \partial j_\tau(x,\xi)$ for all $\xi\in \mathbb{R}^n$ a.e $x\in \Gamma_C$.
	\item[(H8)]
	there exist $c_{1\tau}>0$ such that
	$\|\partial j_{\tau}(x,\xi)\|\le c_{1\tau}$
	for every $\xi\in \mathbb{R}^d$ and a.e. $x\in\Gamma_{C}$.
	\end{itemize}
As an example of such a friction law, one may use a version of the Coulomb law,
\[
|\sigma_\tau |\le \mu p_\nu(u_\nu-g),
\]
where $\mu$ is the coefficient of friction and $\mu p_\nu$ is the friction bound, and
\[
\mbox{if} \  v_\tau\neq 0 \ \mbox{then}  \ \sigma_\tau =- \mu p_\nu(u_\nu-g)\frac{v_\tau}{|v_\tau|}.
\]
That is, frictional resistance takes place only when there is relative motion and then it opposes it.
We can write the condition in a condensed form as an inclusion
\[
\sigma_\tau \in - \mu p_\nu(u_\nu-g)\partial |v_\tau|,
\]
where $\partial |r|$ is the convex subdifferential of $|r|$, i.e.,
\[
\partial |r| =
\begin{cases}
1   & \; r>0, \\
[-1, 1]      & \; r=0, \\
-1   & \; r<0.
\end{cases}
\]
We use the formalism of Clarke subdifferentials in the friction law to account for possible nonmonotonicity 
in the relation between the tangential velocity and the friction force density. This represents the fact 
that kinetic friction can be less than static friction, i.e., a drop of the friction force can occur 
when motion starts. 

Finally, the initial conditions for the displacement, velocity and wear functions are,
\begin{equation}
u(0)=u_0, \ \ u'(0)= v_0, \ \ \theta(0)=\theta_0.\label{intial}
\end{equation}

\section{Variational formulation} 
\label{sect_vf}

 We turn to the variational formulation of problem \eqref{stressdefin}--\eqref{intial}. To that end,  
 we first introduce the concepts that are needed belowl, and then the variational formulation.
  In what follows, $i,j=1,\ldots, d$ everywhere, the summation convention over repeated 
 indices is used, and an index following a comma indicates a partial derivative.

  For a reflexive Banach space $E$, we denote by
  ${\langle \cdot, \cdot \rangle}_{E^* \times E}$ 
  the duality pairing between the dual space $E^*$ and $E$. If $E$ is a Hilbert space, then the 
  scalar product in $E$ is denoted by $(\cdot,\cdot)_E$.  
  Throughout this paper, we denote by $C$ a generic positive constant
  that depends on the problem data and may change its value form line to line. 
  By $|\cdot|$ we denote the Euclidean norm in $\mathbb{R}^d$ or $\mathbb{S}^d$, 
  the space of symmetric $d\times d$ matrices.
  
  To obtain a variational formulation of the model in Section~\ref{Model}, 
  we need the following functional spaces:
  \begin{equation} \nonumber
  \quad H=L^2(\Omega)^d, \ \ 
  V= \{v\in H^1(\Omega)^d\mid  v = 0 ~\textrm{on} ~ \Gamma_D\}. 
  \end{equation}
We know that for $\delta \in (0,\frac{1}{2})$ the embedding $i\colon V\to H^{1-\delta}(\Omega)^d$ 
is compact, and if $\gamma_1\colon H^{1-\delta}(\Omega)^d\to L^2(\Gamma)^d$, 
denotes the trace operator, which is continuous (see, e.g., \cite[Theorem 2.21]{MOSBOOK}), then 
the trace operator $\gamma=\gamma_1 i\colon V\to L^2(\Gamma)^d$
is compact. To simplify slightly the notation, we use $v$ instead of $\gamma v$. 

For a fixed and finite $T>0$ we define the following standard time-dependent spaces:
\begin{align*}
&\mathcal{W}=\{v\in L^2(0,T;V)\mid v'\in L^2(0,T;V^*)\}, \\[2mm]
&\mathcal{W}_\Gamma = \{\theta\in L^2(0,T;H^1(\Gamma_C))\mid \theta'\in L^2(0,T;H^1(\Gamma_C)^*)\}. 
\end{align*}

The Clarke subdifferential of a locally
Lipschitz functional $\varphi \colon \mathbb{R}^d \to \mathbb{R}$  is given by (see~\cite{C})
\[
\partial \varphi (x) = \mbox{conv}\{\lim_{n \to \infty} \nabla \varphi(x_n)  \mid x_n\to x, \nabla \varphi(x_n) \ \textrm{converges}, \ \textrm{and}\ x_n\notin N\cup  N_\varphi\}
\]
where $N_\varphi$ is a set of measure zero, outside of which $\varphi$ is differentiable, and $N$ is 
any set of measure zero. It is possible to generalize the notion of the Clarke subdifferential to functionals 
defined on Banach spaces, cf., \cite{C, MIGDENKBOOK, MOSBOOK}, but for our purposed it is sufficient 
to consider this definition on $\mathbb{R}^d$. 

Now, we define the operators $A, G\colon V\to V^*$ by 
\begin{align*}
\langle {Au,w}\rangle _{V^*\times V}&=\int_\Omega a_{ijkl}\frac{\partial u_k}{\partial x_l}
\frac{\partial w_i}{\partial x_j}\, dx, \quad\quad   u, w \in E,\\
\langle{Gu,\eta}\rangle_{V^*\times V}&=\int_\Omega b_{ijkl}\frac{\partial u_k}
{\partial x_l}\frac{\partial w_i}{\partial x_j}\, dx, \quad\quad  u,w \in E.
\end{align*}

We assume that $f_0(t) \in L^2(\Omega)^d$ and $f_2(t)\in L^2(\Gamma_C)^d$ and this allows 
us to define $f:(0,T)\to V^*$ as
\[
\langle{f(t),\eta}\rangle_{V^*\times V} = \int_\Omega f_0(t) w \, dx + \int_{\Gamma_N} f_2(t) w \, d\Gamma,  \quad\quad  w \in E.
\]

Applying the Green formula and the usual manipulations, we are able to derive the following weak formulation of the problem governed by \eqref{stressdefin}--\eqref{intial}.
 
\begin{problem}
\label{P_V}
Find $u\in L^2(0,T;V)$ with  $v\in \mathcal{W}$ and $\theta\in \mathcal{W}_\Gamma$ such that 
\begin{align}
&\nonumber \langle v'(t),w\rangle_{V^*\times V} + \langle Av(t),w\rangle_{V^*\times V} + \langle Gu(t),w\rangle_{V^*\times V}  \\[2mm]
&\qquad \qquad\qquad \nonumber + \int_{\Gamma_C} h_\nu (u(t))w_\nu \, d\Gamma + \int_{\Gamma_C} h_\tau (u(t),v(t),\theta(t))\xi(t)w_\tau\, d\Gamma  \\[2mm]
& \qquad =\langle f(t),w\rangle_{V^*\times V}, \quad\quad \textrm{for every} \ w\in V, \ \textrm{a.e.} \ t\in (0,T),\\[2mm]
& \xi(t)\in S_{\partial j_\tau}^2 (v_\tau(t))\quad \textrm{a.e.} \ t\in (0,T), \\[2mm]
 & \langle \theta'(t),\eta\rangle_{H^1(\Gamma_C)^*\times H^1(\Gamma_C)} + \kappa ( \nabla_\Gamma \theta(t),\nabla_\Gamma \eta)_{L^2(\Gamma_C)^d} \\[2mm]
&\qquad \nonumber=\int_{\Gamma_C} h_w (u(t),v(t))\eta \, d\Gamma, \quad\quad \textrm{for every} \ \eta \in H^1(\Omega), \ \textrm{a.e.} \ t\in (0,T),\\[2mm]
&\;  u(0)=u_0, \ u'(0)=v_0, \ \theta(0)=\theta_0.
\end{align}
\end{problem}

Here, we used, the notation $v=u'$, i.e.
\begin{equation}\label{hist}
u(t)= u_0 + \int_0^t v(s)\, ds,\qquad  t\in (0,T).
\end{equation}

\noindent
By $\xi(t)\in S_{\partial j_\tau}^2(v_\tau(t))$ we understand a $L^2$-measurable selection 
out of the subdifferential $\partial j$ at $v_\tau$.
We also write, for the sake of simplicity, $H(h)$ as a collection of the hypotheses 
$H(h_w), H(h_\nu)$ and $H(h_\tau)$.

We are now able to state the main theorem of this paper.
  \begin{theorem}
  \label{mainth}
  Assume that $u_0\in V, v_0\in H$, $\theta_0\in L^2(\Gamma_C)$, $f_0\in L^2(0,T;V^*), \ f_2\in L^2(0,T;L^2(\Gamma_N)^d)$, and $\kappa >0$. Under hypotheses (H1)--(H8) there exists a solution to Problem~\ref{P_V}. 
  \end{theorem}

We conclude that the model  \eqref{stressdefin}--\eqref{intial} has a weak or variational solution.
The uniqueness of the solution remains an unresolved question.

\section{Proof of Theorem~\ref{mainth}}
\label{sec:ex}

In this section we prove the existence theorem. The idea of the proof is as follows. First, we decouple the
coupled Problem~\ref{P_V} by replacing  the coupling terms with given functions and introduce 
truncation operators. We obtain the existence of solutions for the decoupled and truncated problems 
independently. Then, we apply the Kakutani--Ky Fan--Glicksberg fixed point theorem to show the existence 
result for the original problem. Finally, we pass to the limit with the truncation parameter. In the 
proof we always assume (H1)--(H8), and that $u_0\in V, v_0\in H$, $\theta_0\in L^2(\Gamma_C)$, $f_0\in L^2(0,T;V^*), \ f_2\in L^2(0,T;L^2(\Gamma_N)^d)$, and $\kappa >0$, so we do not repeat 
these assumptions in the auxiliary lemmas below.

We start by recalling the fixed point theorem.

\begin{theorem}[Kakutani--Ky Fan--Glicksberg]\label{Kakutani}
	Let $S\subset E$ be a non-empty, compact, and convex set, where $E$ is a locally convex Hausdorff topological vector space. Let the set-valued function $\varphi\colon S\to 2^S$ have non-empty, convex values, and let $\textrm{Gr}(\varphi) = \{ (x,y)\in S\, |\ y\in \varphi(x) \} $ be a closed set in the product topology of $E\times E$. Then, the set $\{x\in S\mid x\in \varphi(x)\}$ of fixed points of $\varphi$ is non-empty and compact.
	\end{theorem}

Next, for $l>0$, we define truncation operators $N_l\colon \mathbb{R}^d\to \mathbb{R}^d$ 
and  $M_l\colon \mathbb{R}\to \mathbb{R}$ by 
\vskip4pt

$\begin{array}{ll}
\qquad N_l(x)=\begin{cases} x, \qquad |x|\le l, \\
\frac{x}{|x|} l, \quad |x|>l.
\end{cases} & M_l(x)=\begin{cases} x, \qquad |x|\le l, \\
\frac{x}{|x|} l, \quad |x|>l.
\end{cases}
\end{array}$
\medskip

The following lemma is straightforward to show, and we present the proof for the sake 
of completeness.

 \begin{lemma}\label{Liptrunc}
 	Truncation operators $N_l$ and $M_l$ are Lipschitz continuous with a constant 1.
 \end{lemma}
\begin{proof}
	We present the proof only for $N_l$. Let $x,y\in \mathbb{R}^d$ and we 
	consider the three cases: $|x|,|y|\le l$, $|x|> l, |y| \le l$, and $|x|,|y|>l$. 
	In the fist case, we immediately obtain the result. In the second case, we 
	calculate the inner products in $\mathbb{R}^d$,
\begin{eqnarray}
&& \nonumber( N_l(x)-N_l(y), N_l(x)-N_l(y) ) = \left( l\frac{x}{|x|}-y,l\frac{x}{|x|}-y \right) =l^2-\frac{2 l( x,y ) }{|x|} + |y|^2 \\[2mm]
&&\quad \nonumber \le l^2 - |x|^2  +2 ( x,y ) \frac{|x|-l}{|x|} + |x-y|^2 \le | x-y|^2. 
\end{eqnarray}
In the last case,
\begin{eqnarray}
&& \nonumber ( N_l(x)-N_l(y), N_l(x)-N_l(y) ) = \left( l\frac{x}{|x|}-l\frac{y}{|y|}, l\frac{x}{|x|}-l\frac{y}{|y|} \right) \\[2mm]
&&\quad 
\nonumber= 2l^2 -2(x,y)\frac{l^2}{|x||y|} \le |x-y|^2.
\end{eqnarray} 
	\end{proof}

 Now we fix $l>0$ (large), choose the functions $\bar{v}\in \mathcal{W}, \ \bar{\xi} \in L^2(0,T;L^2(\Gamma_C)^d)$ and  $\bar{\theta} \in \mathcal{W}_\Gamma$, and let $\bar{u}$ given by \eqref{hist} using $\bar{v}$. 
 Consider now the following two auxiliary problems.
 
 \begin{problem}\label{P_aux1}
 Find the velocity field $v\in \mathcal{W}$ such that 
\begin{align}
 &\label{aux_1} \langle v'(t),w\rangle_{V^*\times V} +\langle Av(t),w\rangle_{V^*\times V} + \langle Gu(t),w\rangle_{V^*\times V} \\[2mm] \nonumber
& \qquad \qquad + \int_{\Gamma_C} h_\nu (N_l\bar{u}(t))w_\nu \, d\Gamma   + \int_{\Gamma_C} h_\tau (N_l\bar{u}(t),N_l\bar{v}(t),M_l\bar{\theta}(t))\bar{\xi}(t)w_\tau\, d\Gamma \\[2mm]
&\qquad  =\langle f(t),w\rangle_{V^*\times V} \quad \textrm{for every} \ w\in V \ \textrm{a.e.} \ t\in (0,T),\nonumber \\[2mm]
&v(0)=v_0\nonumber.
 \end{align}
 \end{problem}

 \begin{problem} \label{P_aux2}
  Find the wear function $\theta\in \mathcal{W}_\Gamma$ such that
  \begin{align}
\label{aux_2} & \langle \theta'(t),\eta\rangle_{H^1(\Gamma_C)^*\times H^1(\Gamma_C)} + \kappa ( \nabla_\Gamma \theta(t),\nabla_\Gamma \eta )_{L^2(\Gamma_C)^d} \\[2mm]
 &\qquad =\int_{\Gamma_C} h_w (N_l\bar{u}(t),N_l\bar{v}(t))\eta \, d\Gamma \quad \textrm{for every} \ \eta \in H^1(\Gamma_C) \ \textrm{a.e.} \ t\in (0,T),\nonumber\\[2mm]
 & \theta(0)=\theta_0.\nonumber
 \end{align}
 \end{problem}
 
We note that by using the given functions and the truncations, the two problems
are uncoupled.

 \begin{lemma}\label{exist_1}
 	There exists a unique solution to Problem~\ref{P_aux1} 
 \end{lemma}
 \begin{proof} For the proof of Lemma we refer to \cite{MOSBOOK}. 
 \end{proof}

 \begin{lemma}\label{exist_2}
There exists a unique solution to Problem~\ref{P_aux2}.
 \end{lemma}
\begin{proof} For the proof, we refer to classical results on parabolic problems, see, e.g., 
\cite{Lions}. 
\end{proof}
 
 In the next step, we introduce the following coupled, but still truncated problem.
 
 \begin{problem}\label{P_trunc}
 Find $v\in \mathcal{W}$  with $\xi\in L^2(0,T;L^2(\Gamma_C)^d)$ and $\theta \in \mathcal{W}_\Gamma$, 
 such that
 
 \begin{align}
&\nonumber  \langle v'(t),w\rangle_{V^*\times V} +\langle Av(t),w\rangle_{V^*\times V} + \langle Gu(t),w\rangle_{V^*\times V}  \\[2mm]
&\qquad \qquad \qquad \nonumber  + \int_{\Gamma_C} h_\nu (N_lu(t))w_\nu \, d\Gamma + \int_{\Gamma_C} h_\tau (N_lu(t),N_lv(t),M_l\theta(t))\xi(t)w_\tau\, d\Gamma  \\[2mm]
&\qquad \label{est_p1} = \langle f(t),w\rangle_{V^*\times V} \quad 
\textrm{for every} \ w\in V, \ \textrm{a.e.} \ t\in (0,T),
\end{align}

 \begin{align}
&\label{est_p2}  \xi(t)\in S_{\partial j_\tau}^2 (v_\tau(t))\quad \textrm{a.e.} \ t\in (0,T), \\[2mm]
&\nonumber   \langle \theta'(t),\eta\rangle_{Y^*\times Y} + \kappa \langle \nabla_\Gamma \theta(t),\nabla_\Gamma \eta\rangle_{L^2(\Gamma_C);\mathbb{R}^d} \\[2mm]
&\qquad  =\int_{\Gamma_C} h_w (N_lu(t),N_lv(t))\eta \, d\Gamma \quad 
 \label{est_p3} 
 \textrm{for every} \ \eta \in Y, \ \textrm{a.e.} \ t\in (0,T),\\[2mm]
 & v(0)=v_0, \ \theta(0)=\theta_0.
 \end{align}
 
 \end{problem}
 
We now show the existence of a solution to Problem~\ref{P_trunc} by using 
Lemmas~\ref{exist_1} and~\ref{exist_2} and the fixed point theorem, Theorem ~\ref{Kakutani}. 

In what follows, we check all the assumption of Theoremt~\ref{Kakutani}, and summarize the steps in the lemmas. 
First, we derive  the necessary a-priori estimates.

\begin{lemma}\label{radius}
Let $v$ and $\theta$ be the solutions of Problems~\ref{P_aux1} and~\ref{P_aux2}, respectively. Then,
the following estimates hold:
\begin{eqnarray}\label{rad_1}
&& \|v\|_{\mathcal{W}}^2\le  C_1\left(1+\|u_0\|_V^2 + \|v_0\|_H^2 + \|\bar{\xi}\|_{L^2(0,T;L^2(\Gamma_C)^d}^2\right), \\[2mm]
&& \|\theta\|_{\mathcal{W}_{\Gamma}}^2 \le  C\left(1+\|\theta_0\|_{L^2(\Gamma_C)}^2\right).\label{rad_2}
\end{eqnarray} 
Moreover, there exists $\xi \in L^2(0,T;L^2(\Gamma_C)^d)$ that satisfies 
\[
\xi(t)\in S_{\partial j_\tau}^2 (\overline{v}_\tau(t))\quad \textrm{a.e.} \ t\in (0,T),
\]
and the bound 
\begin{equation}
\|\xi\|_{L^2(0,T;L^2(\Gamma_C)^d)}^2 \le C_3. \label{rad_3}
\end{equation}
The constants $C_1,C_2,C_3>0$, depend only  $\Omega$, $T$, $f$, the constants present in (H1)-(H8), and $l$.
\end{lemma}

\begin{proof} 
We choose $w=v(t)$ in (\ref{aux_1}) and then it follows from (H3), (H5), (H6) and the Cauchy inequality 
with $\varepsilon>0$ that for $t\in (0,T)$, 
\begin{eqnarray}
\nonumber &&\frac{1}{2}\frac{d}{dt} \|v(t)\|_H^2 +\alpha \|v(t)\|_V^2 +\frac{1}{2}\frac{d}{dt} 
\langle Gu(t),u(t)\rangle_{V^*\times V}  \\[2mm]
\nonumber
 &&\qquad \le C\left(1+{\varepsilon}\|v(t)\|_V^2 + C(\varepsilon)\|v(t)\|_{V^*}^2 +  
 C(\varepsilon)\|\bar{\xi}(t)\|_{L^2(\Gamma_C)^d}^2\right).\label{est_4}
\end{eqnarray}
Integrating (\ref{est_4}) over $(0,t)$ for $t\in (0,T)$ and choosing appropriate value of $\varepsilon$ yields 
\begin{eqnarray} &&\nonumber\|v\|_{L^\infty(0,T;H)}^2 + \alpha\|v(t)\|_{L^2(0,T;V)}^2 \\[2mm]
&&\qquad \qquad \le C\left(1+\|u_0\|_V^2 + \|v_0\|_H^2 + \|\bar{\xi}\|_{L^2(0,T;L^2(\Gamma_C)^d)}^2\right).\label{est_v}
\end{eqnarray}
Next, we choose $\eta=\theta(t)$ in (\ref{aux_2}), and then  it follows from (H4) that
\begin{eqnarray}
\frac{1}{2}\frac{d}{dt} \|\theta(t)\|_{L^2(\Gamma)}^2 + \kappa \|\nabla_\Gamma \theta(t)\|_{L^2(\Gamma_C)^d}^2 \le C\left(1+\|\theta(t)\|_{L^2(\Gamma_C)^2}^2\right).\label{est_5}
\end{eqnarray}
Again, integrating (\ref{est_5}) over $(0,t)$ for $t\in (0,T)$ we obtain
\begin{equation}\label{est_1}
\|\theta(t)\|_{L^2(\Gamma_C)}^2 + \kappa\|\nabla_\Gamma \theta\|_{L^2(0,t;L^2(\Gamma_C)^d)}^2 \le C\left(1 + \|\theta_0\|_{L^2(\Gamma_C)}^2 + \|\theta\|_{L^2(0,t;L^2(\Gamma_C))}^2\right).
\end{equation}
Using the Gronwall inequality we get for $t\in (0,T)$
\begin{equation}\label{est_2}
\|\theta(t)\|_{L^2(\Gamma_C)}^2 \le C\left(1+\|\theta_0\|_{L^2(\Gamma_C)}^2\right).
\end{equation}
Combining (\ref{est_1}) and (\ref{est_2}) it follows that
\begin{equation}\label{est_theta1}
\|\nabla_\Gamma \theta\|_{L^2(0,T;L^2(\Gamma_C)^d)}\le C\left(1+\|\theta_0\|_{L^2(\Gamma_C)}^2\right), 
\end{equation}
hence from \eqref{est_2} and \eqref{est_theta1} we conclude that
\begin{equation}
\|\theta\|_{L^2(0,T;H^1(\Gamma_C))}^2\le C\left(1+\|\theta_0\|_{L^2(\Gamma_C)}^2\right).\label{est_t2}
\end{equation}
Straightforward manipulations  using the estimates \eqref{est_v} and \eqref{est_t2} and \eqref{aux_1} and \eqref{aux_2} 
lead to the following bounds on $v'$ and $\theta'$, 
\begin{eqnarray}
&& \|v'\|_{L^2(0,T;V^*)}^2 \le C\left(1+\|u_0\|_V^2 + \|v_0\|_H^2 + \|\bar{\xi}\|_{L^2(0,T;L^2(\Gamma_C)^d)}^2\right), \label{est_v'}\\[2mm]
&& \|\theta'\|_{L^2(0,T;H^1(\Omega)^*)}^2 \le C\left(1+\|\theta_0\|_{L^2(\Gamma_C)}^2\right).\label{est_t'}
\end{eqnarray}

Now, we need to show the existence of $\xi \in L^2((0,T)\times \Gamma_C)^d$ such that $\xi(t) \in S^2_{\partial j_\tau}
(v_\tau(t))$ for a.e. $t\in (0,T)$. To show this it is sufficient to prove the existence of a measurable 
selection, since the integrability, as well as the bound \eqref{rad_3}  follow from (H8). However, the existence of a 
measurable selection of the subdifferential follows from  \cite[Theorem 5.6.39]{MIGDENKBOOK}, as the 
Clarke subdifferential  of the locally Lipschitz integral functional
\[
J: L^2(0,T;L^2(\Gamma_C)^d)\to \mathbb{R},
\]
defined by 
\[
J(v) = \int_0^T \int_{\Gamma_C} j_\tau(v(x,t))\, d\Gamma dt,
\]
is nonempty and its elements are measurable on the one-hand, and on the other-hand they 
are selections of the multifunction $\partial j(v(x,t))$, see also \cite{GASKAL}.

Therefore,  \eqref{est_v}, \eqref{est_t2}--\eqref{est_t'} and (H8) imply that, for some positive
constants $C_1,C_2$ and $C_3$, the following estimates:
\begin{eqnarray*}\label{est_3}
&& \|v\|_{\mathcal{W}}^2\le  C_1\left(1+\|u_0\|_V^2 + \|v_0\|_H^2 + \|\bar{\xi}\|_{L^2(0,T;L^2(\Gamma_C)^d)}^2\right), \\[2mm]
&& \|\theta\|_{\mathcal{W}_{\Gamma}}^2 \le  C_2\left(1+\|\theta_0\|_{L^2(\Gamma_C)}^2\right), \\[2mm]
&& \|\xi\|_{L^2(0,T;L^2(\Gamma_C)^d)}^2\le C_3.
\end{eqnarray*}
This completes the proof of the Lemma~\ref{radius}. 
\end{proof}

Next, we define the space $Z= \mathcal{W}\times\mathcal{W}_\Gamma\times L^2(0,T;L^2(\Gamma_C)^d)$ 
and consider the solution operator $\Lambda\colon Z \to 2^Z$, which assigns to a triple 
$(\bar{v},\bar{\theta},\bar{\xi})$ a triple $(v,\theta,\xi)$, where $v$ and $\theta$ are the solutions of Problems~\ref{P_aux1} and \ref{P_aux2}, respectively, and $\xi$ is a $L^2$-measurable selection out
of $\partial j(\overline{v}_\tau)$. We have the following lemma.
\begin{lemma}\label{balls}
There exist positive constants $R_1,R_2$ and $R_3$ such that 
 $\Lambda(B)\subset 2^{B}$, where the set $B=B(R_1,R_2,R_3)$ is given by
\[B(R_1,R_2,R_3)= \{(v,\theta,\xi) \in Z \mid \|v\|_\mathcal{W}\le R_1, 
\|\theta\|_{\mathcal{W}_\Gamma}\le R_2, \|\xi\|_{L^2(0,T;L^2(\Gamma_C)^d)}\le R_3\}.\]
\end{lemma}

\begin{proof}
We use the estimates above and choose 
\begin{equation}
R_3= C_3, R_2=C_2\left(1+\|\theta_0\|_{L^2(\Gamma_C)}^2\right),\;  R_1=C_1
\left(1+\|u_0\|_V^2 + \|v_0\|_H^2 + R_3\right).
\end{equation}
Now, the assertion of the Lemma  follows  from Lemma~\ref{radius}.
\end{proof}
 
\begin{lemma}\label{convex}
$\Lambda$ has non-empty and convex values.
\end{lemma}
\begin{proof}
The result follows from the convexity in the definition of the Clarke subdifferential, 
see eg. \cite{C}, the existence of $\xi$ given in Lemma \ref{radius}, and the existence and 
uniqueness of $v$ and $\theta$ established in Lemmas \ref{exist_1} and \ref{exist_2}.
\end{proof}

\begin{lemma}\label{graph}
$Gr(\Lambda)$ is sequentially closed in $(w-Z)\times(w-Z)$ topology.
\end{lemma}
\begin{proof}
We choose three sequences such that $\bar{v}_n\to \bar{v}$ weakly in 
$\mathcal{W}$, $\bar{\theta}_n\to \bar{\theta}$ 
weakly in $\mathcal{W}_\Gamma$ and $\bar{\xi}_n \to \bar{\xi}$ weakly in $L^2(0,T;L^2(\Gamma_C)^d)$. 
Define $v_n, \theta_n$ and $\xi_n$ as, respectively, the solutions of Problems \ref{P_aux1} 
and \ref{P_aux2} corresponding to $\bar{v}_n, \bar{\theta}_n, \bar{\xi}_n$, and the $L^2$ 
selection of $\partial j_\tau(\bar{v}_{n\tau}(x,t))$. Assume that $v_n\to v$ weakly in 
$\mathcal{W}$, $\theta_n\to \theta$ weakly in $\mathcal{W}_\Gamma$ and $\xi_n \to \xi$ 
weakly in $L^2(0,T;L^2(\Gamma_C)^d)$. We need to show that 
$v$ and $\theta$  are the solutions of Problems \ref{P_aux1} and \ref{P_aux2} that 
correspond to $\bar{v}, \bar{\theta}$ and $\bar{\xi}$, and that $\xi$ is the $L^2$ 
selection of $\partial j_\tau(\overline{v}_{\tau}(x,t))$.

First observe that the compactness of the embedding $i: V\to H^{1-\delta}(\Omega)^d$ together 
with the Aubin--Lions lemma imply that the tangential components of the traces satisfy
\[
\overline{v}_{n\tau} \to \overline{v}_\tau \ \textrm{strongly in}\ L^2(0,T;L^2(\Gamma_C)^d).
\] 
Since $\xi_n\to \xi$ weakly in $L^2(0,T;L^2(\Gamma_{C})^d)$ and $\xi_n$ is a selection out of $\partial j_\tau(\overline{v}_{n\tau}(x,t))$, a standard argument based on the Aubin--Cellina convergence 
theorem, \cite[Theorem 7.2.2]{Frankowska}, implies that $\xi$ is a selection out of  
$\partial j_\tau(\overline{v}_{\tau}(x,t))$. 

To show that  
$v$ and $\theta$  are the solutions of Problems \ref{P_aux1} and \ref{P_aux2}
corresponding to $\bar{v}, \bar{\theta}$ and $\bar{\eta}$, we need to write \eqref{aux_1} 
and \eqref{aux_2} for $v_n$ and $\theta_n$, and then pass to the limit $n \to \infty$. 
It is clear that $\theta_n(0) \to \theta(0)$ weakly in $L^2(\Gamma_C)$ and $v_n(0) \to v(0)$ 
weakly in $H$, which implies that $\theta$ and $v$ satisfy the same initial conditions as 
$\theta_n, v_n$. Moreover, the following hold,
\begin{align*}
&\langle v_n'(t),w\rangle_{V^*\times V} +\langle Av_n(t),w\rangle_{V^*\times V} + \langle Gu_n(t),w\rangle_{V^*\times V} \\[2mm] \nonumber
& \qquad \qquad + \int_{\Gamma_C} h_\nu (N_l\bar{u}_n(t))w_\nu \, d\Gamma   + \int_{\Gamma_C} h_\tau (N_l\bar{u}_n(t),N_l\bar{v}_n(t),M_l\bar{\theta}_n(t))\bar{\xi}_n(t)w_\tau\, d\Gamma \\[2mm]
&\qquad  =\langle f(t),w\rangle_{V^*\times V} \quad  \forall w\in V \ \textrm{a.e.} \ t\in (0,T),\nonumber \\[2mm]
 & \langle \theta_n'(t),\eta\rangle_{H^1(\Gamma_C)^*\times H^1(\Gamma_C)} + \kappa ( \nabla_\Gamma \theta_n(t),\nabla_\Gamma \eta )_{L^2(\Gamma_C)^d} \\[2mm]
&\qquad =\int_{\Gamma_C} h_w (N_l\bar{u}_n(t),N_l\bar{v}_n(t))\eta \, d\Gamma \quad  \forall  \eta \in H^1(\Gamma_C) \ \textrm{a.e.} \ t\in (0,T).\nonumber
\end{align*}
To show the weak sequential closedness of $Gr(\Lambda)$, we show the convergence of all the terms. 
The Aubin--Lions lemma implies  
 \begin{align}
 & v_n\to v\  \textrm{strongly in}  \ L^2(0,T;H) \; \textrm{ and }\;  \theta_n\to \theta \ \textrm{strongly in}\ L^2(0,T;L^2(\Gamma_C)),\label{est_6}\\
 &\overline{v}_{n} \to \overline{v} \; \textrm{ strongly in }\; L^2(0,T;L^2(\Gamma_C)^d), \overline{u}_{n} \to \overline{u} \ \textrm{strongly in}\ L^2(0,T;L^2(\Gamma_C)^d),\label{conv1}\\
 &\bar{\theta}_n\to \bar{\theta} \ \textrm{strongly in}\ L^2(0,T;L^2(\Gamma_C)).\label{conv2}
 \end{align}
 \noindent
Now, Lemma~\ref{balls}, (\ref{est_3}) and (\ref{est_6}), and the linearity of the duality pairings and the linearity and boundedness of operators $A$ and $G$, we obtain  that
\begin{eqnarray*}
&& \int_0^T \langle v_n'(t),w(t)\rangle_{V^*\times V} \, dt\to  \int_0^T\langle v'(t),w(t)\rangle_{V^*\times V}\, dt, \\
&&  \int_0^T\langle Av_n(t),w(t)\rangle_{V^*\times V} \, dt\to  \int_0^T\langle Av(t),w(t)\rangle_{V^*\times V}\, dt, \\
&&  \int_0^T\langle Gu_n(t),w(t)\rangle_{V^*\times V} \, dt\to  \int_0^T\langle Gu(t),w(t)\rangle_{V^*\times V}\, dt, \\
&&  \int_0^T\langle \theta_n'(t),\eta(t)\rangle_{H^1(\Gamma_C)^*\times H^1(\Gamma_C)} \, dt\to  \int_0^T\langle \theta(t),\eta(t)\rangle_{H^1(\Gamma_C)^*\times H^1(\Gamma_C)}\, dt,\\
&&  \int_0^T\kappa (\nabla_\Gamma \theta_n(s),\nabla_\Gamma \eta(t))_{L^2(\Gamma_C)^d} \, dt\to  \int_0^T \kappa(\nabla_\Gamma\theta(s),\nabla_\Gamma\eta(t) )_{L^2(\Gamma_C)^d}  \, dt,
\end{eqnarray*}
for every $w\in L^2(0,T;V)$, and every $\eta\in L^2(0,T;H^1(\Omega))$, as $n\to \infty$.
Next, we deal with the boundary integrals. To simplify the presentation, we omit the 
time dependence of the functions. We write
\begin{eqnarray}
\nonumber &&\int_0^T \int_{\Gamma_C} h_\tau(\bar{u}_n,\bar{v}_n,\bar{\theta}_n)\bar{\xi}_n w_\tau \, d\Gamma \, ds= \int_0^T \int_{\Gamma_C} h_\tau(\bar{u}_n,\bar{v}_n,\bar{\theta}_n)\bar{\xi}_n w_\tau \, d\Gamma\,  ds \\[2mm] 
&&\nonumber-\int_0^T \int_{\Gamma_C} h_\tau(\bar{u},\bar{v},\bar{\theta})\bar{\xi}_n w_\tau  \, d\Gamma\,  ds
 + \int_0^T \int_{\Gamma_C} h_\tau(\bar{u},\bar{v},\bar{\theta})\bar{\xi}_n w_\tau \, d\Gamma \, ds.
\end{eqnarray}
The weak convergence $\bar{\xi}_n \to \bar{\xi}$ in $L^2(0,T;L^2(\Gamma_C)^d)$ implies 
\begin{equation}
\int_0^T \int_{\Gamma_C} h_\tau(\bar{u},\bar{v},\bar{\theta})\bar{\xi}_n w_\tau\, d\Gamma ds \to  \int_0^T \int_{\Gamma_C} h_\tau(\bar{u},\bar{v},\bar{\theta})\bar{\xi} w_\tau d\Gamma ds.\label{h_tauconv1}
\end{equation}
Moreover, by the continuity of $h_\tau$, the strong convergences \eqref{conv1} and \eqref{conv2}
and the Lebesgue dominated convergence theorem, we find
 \begin{eqnarray}
 \nonumber&&\int_0^T \int_{\Gamma_C} h_\tau(\bar{u}_n,\bar{v}_n,\bar{\theta}_n)\bar{\xi}_n w_\tau  d\Gamma ds-\int_0^T \int_{\Gamma_C} h_\tau(\bar{u},\bar{v},\bar{\theta})\bar{\xi}_n w_\tau  d\Gamma ds \ \\[2mm]
&& \; \le\|\bar{\xi}_n\|_{L^2(0,T;L^2(\Gamma_C)^d)} \left(\int_0^T\int_{\Gamma_C} |w_\tau|^2 (h_\tau(\bar{u}_n,\bar{v}_n,\bar{\theta}_n)- h_\tau (\bar{u},\bar{v},\bar{\theta}))\, d\Gamma ds\right)^{1/2}, \label{h_tauconv2}
\end{eqnarray}
\noindent
where the last term converges to zero as $n\to \infty$. Hence, \eqref{h_tauconv1} 
and \eqref{h_tauconv2} yield that as $n\to \infty$,
\begin{equation}
\int_{\Gamma_C} h_\tau(\bar{u}_n,\bar{v}_n,\bar{\theta}_n)\bar{\xi}_n w_\tau \, d\Gamma \, ds \to \int_0^T \int_{\Gamma_C} h_\tau(\bar{u},\bar{v},\bar{\theta})\bar{\xi} w_\tau \,d\Gamma\,  ds.\nonumber 
\end{equation}

\noindent
By the direct application of the Lebesgue dominated convergence theorem and the
continuity of  $h_\nu$ and $h_w$, we find
\begin{eqnarray}
&&\int_0^T \int_{\Gamma_C} h_\nu(\bar{u}_n,)w_\nu \, d\Gamma \, ds \to \int_0^T \int_{\Gamma_C} h_\nu(\bar{u}) w_\nu d\Gamma ds,\nonumber \\[2mm]
&&\int_{\Gamma_C} h_w(\bar{u}_n,\bar{v}_n)\eta \, d\Gamma \, ds \to \int_0^T \int_{\Gamma_C} h_w(\bar{u},\bar{v})\eta \, d\Gamma \, ds, \nonumber 
\end{eqnarray}
\noindent
as $n\to \infty$. This completees proof of the lemma.
\end{proof}

The next step is essentially the last one.

\begin{lemma}
\label{fixed}
The operator $\Lambda$ has a fixed point.
\end{lemma}
\begin{proof}
Consider $\Lambda|_{B(R_1,R_2,R_3)}$, where $B(R_1,R_2,R_3)$ is given by Lemma~\ref{balls}. It follows 
from the lemma that $\Lambda(B(R_1,R_2,R_3))\subset 2^{B(R_1,R_2,R_3)}$. Then, Lemma \ref{convex} s
hows that this mapping has nonempty and convex values. From Lemma~\ref{graph} we deduce that 
$Gr(\Lambda|_{B(R_1,R_2,R_3)})$ is sequentially closed in the $(w-Z)\times (w-Z)$ topology. Since the 
topology is weak, we need the following argument to show that this set is closed.  
But, $Gr(\Lambda|_{B(R_1,R_2,R_3)})\subset B(R_1,R_2,R_3)\times B(R_1,R_2,R_3)$, which is bounded, 
 closed and convex in the reflexive space $Z\times Z$,therefore, $Gr(\Lambda|_{B(R_1,R_2,R_3)})$ is sequentially compact, and so $(w-Z)\times (w-Z)$ is compact and $(w-Z)\times (w-Z)$ is closed. Taking into account Lemma~\ref{convex}, the assertion of the lemma follows now directly from Theorem~\ref{Kakutani}.
\end{proof}

We have shown that all the assumptions of the fixed-point theorem, Theorem~\ref{Kakutani}, hold true
 and that establishes the following theorem, which guarantees the existence of a solution of the truncated problem. 

\begin{theorem}\label{th_trunc}
There exists a solution to Problem~\ref{P_trunc}.
 \end{theorem}

The last step in the proof of our main theorem is to show that we can remove the truncation operators from Problem~\ref{P_trunc}.

\begin{proof}[Proof of Theorem~\ref{mainth}]  We need to obtain the relevant estimates on a 
solution $v,\theta$ of Problem~\ref{P_trunc} that are independent of the truncation parameter $l$. 
To that end, we choose $\eta=v(t)$ in (\ref{est_p1}) and using again the Cauchy inequality with $\varepsilon>0$,
we obtain
\begin{eqnarray}
\nonumber&& \frac{1}{2}\frac{d}{dt} \|v(t)\|_H^2 + \alpha \|v(t)\|_V^2 + \frac{1}{2}\frac{d}{dt} \langle Gu(t),u(t)\rangle_{V^*\times V}  \\
\nonumber&& \qquad \qquad  + \int_{\Gamma_C} h_\nu(N_l u(t))v_\nu(t)\, d\Gamma  +\int_{\Gamma_C} h_\tau(N_lu(t),N_lv(t),M_l\theta(t))\xi(t)v_\tau(t)\, d\Gamma \\ 
&& \qquad \le C(\varepsilon)\|f(t)\|_{V^*}^2 + \varepsilon \|v(t)\|_V^2, \label{est_7}
\end{eqnarray}
for $t\in (0,T)$. The hypotheses (H5)-(H7)  and an appropriate choice of $\varepsilon>0$
in (\ref{est_7}) yields
\begin{eqnarray}
\nonumber&&  \frac{d}{dt} \|v(t)\|_H^2 + \alpha\|v(t)\|_V^2 + \frac{d}{dt} \langle Gu(t),u(t)\rangle_{V^*\times V }   \\
&& \qquad \le C\Big(1 + \|f(t)\|^2_{V^*}+ \int_{\Gamma_C} |v(t)|\, d\Gamma+ \int_{\Gamma_C} |u(t)||v(t)|\, d\Gamma)\label{est_8}
\end{eqnarray}
\noindent
for $t\in (0,T)$. 
Straightforward manipulations that use the Cauchy inequality with $\varepsilon$ again, 
the fact that $u(t) = u_0 + \int_0^t v(t)\, dt$ and the inequality in {\cite{MIGDENKBOOK}, Lemma 8.4.12}
show that
\[
\|v\|_{L^2(\Gamma_C)^2}^2 \leq \varepsilon \|v\|_{V}^2 + C(\varepsilon)\|v\|_H^2, \label{erhling}
\]
which leads to 
\begin{eqnarray}
	&& \nonumber \frac{d}{dt} \left(  \|v(t)\|_H^2 + \langle Gu(t),u(t)\rangle_{V^*\times V}  \right) + \alpha \|v(t)\|_V^2 \\[2mm]
	&& \qquad \le C\Big(1  + \|f(t)\|^2_{V^*} + \|u_0\|_V^2 + \int_0^t \|v(s)\|_H^2\ ds\Big). \label{est_on_v}
\end{eqnarray}
Integrating \eqref{est_on_v} over$(0,t),~ t \in (0,T)$ and using (H1) and (H3), we get

\begin{equation}
\|v(t)\|_H^2 +   \|v\|_{L^2(0,t;V)}^2  \le C \left(1+ \|u_0\|_V^2 + \|v_0\|_H^2 + 
\int_0^t \|v(s)\|_{H}^2 \, ds\right).\label{est_v_int}
\end{equation}
By the Gronwall inequality applied to $\|v(t)\|_H^2$ we find 
\begin{equation} \|v(t)\|_H^2\le C, \label{est_12a}
\end{equation}
for $t\in (0,T)$. Applying \eqref{est_12a} to \eqref{est_v_int} we obtain
\begin{equation}
\|v\|_{L^2(0,T;V)}^2  \le C. \label{est_13a}
\end{equation}
Choosing $\eta=\theta(t)$ in (\ref{est_p3}), and applying (H4) for $t\in (0,T)$, leads to the estimate
\begin{eqnarray}
\nonumber&& \frac{d}{dt}\|\theta(t)\|_{L^2(\Gamma_C)}^2 + \kappa \|\nabla_\Gamma \theta(t)\|_{L^2(\Gamma_C)^d}^2  \le 
C\Big(1+\int_{\Gamma_C} |u(t)|^2|\theta(t)| \, d\Gamma \\[2mm] 
&& \qquad \qquad +\int_{\Gamma_C} |v(t)|^2 |\theta(t)| \, d\Gamma  + \int_{\Gamma_C} |\theta(t)|\, d\Gamma\Big).\label{est9a}
\end{eqnarray}
Next, by using the continuous embedding $H^1(\Omega) \to L^4(\partial\Omega)$, we find
\begin{eqnarray}
&& \nonumber \frac{d}{dt}\|\theta(t)\|_{L^2(\Gamma_C)}^2 + \|\nabla_\Gamma \theta(t)\|_{L^2(\Gamma_C)^d}^2\le C \Big(1+ \|u(t)\|_V^2\|\theta(t)\|_{L^2(\Gamma_C)} \\[2mm] 
&&
\quad \quad + \|v(t)\|_V^2\|\theta(t)\|_{L^2(\Gamma_C)} + \|\theta(t)\|_{L^2(\Gamma_C)}^2\Big). \label{est10a}
\end{eqnarray}
Integrating \eqref{est10a} over $(0,t), ~t\in (0,T)$, we find
\begin{eqnarray}
&& \|\theta(t)\|_{L^2(\Gamma_C)}^2 + \|\nabla \theta(t)\|_{L^2(0,t;{L^2(\Gamma_C)}^d)}^2 \le C \Big(1+ \int_0^t \|u(s)\|_V^2\|\theta(s)\|_{L^2(\Gamma_C)}\, ds \\[2mm] 
&&
\quad \quad + \int_0^t \|v(s)\|_V^2\|\theta(s)\|_{L^2(\Gamma_C)}\, ds + \int_0^t \|\theta(s)\|_{L^2(\Gamma_C)}^2 \, ds \Big). \label{est11a}
\end{eqnarray}
Using a nonlinear version of the Gronwall inequality (\cite[p.360]{MITRIN}), we conclude that
\begin{equation}
\|\theta(t)\|_{L^2(\Gamma_C)}^2 \le C, \label{est_theta}
\end{equation}
for $t\in (0,T)$ and, consequently, applying \eqref{est_theta} to \eqref{est11a} we have 
\begin{equation}
\| \theta\|_{L^2(0,T;{H^1(\Gamma_C)})}^2 \le C. \label{est_nabla}
\end{equation}
The previous estimates imply the bound
\begin{equation}
\|v'\|_{L^2(0,T;V^*)}^2+ \|\theta'\|_{L^2(0,T;H^1(\Omega)^*)}^2 \le C, \label{est_derivatives}
\end{equation}
and so we conclude from \eqref{est_12a}, \eqref{est_13a} and \eqref{est_theta}--\eqref{est_derivatives} that
\begin{equation}
\|v\|_{\mathcal{W}}^2 + \|\theta\|_{\mathcal{W}_\Gamma}^2 \le C, \label{est_final}
\end{equation}
where $C$ is independent of $l$.  This estimate is crucial for the proof of the theorem.
\medskip

Now, let $(v_n, \theta_n, \xi_n)$ be a solution of Problem \ref{P_trunc} with the truncation 
constant $l=n$. Then, \eqref{est_final} and the Aubin--Lions lemma imply that there is a subsequences 
such that $v_n\to v$ strongly in $L^2(0,T;H)$ and in $L^2(0,T;L^2(\Gamma_C)^d)$, and 
$\theta_n\to \theta$ strongly in $L^2(0,T;L^2(\Gamma_C))$. Passing to the limit with the multivalued 
term follows exactly as in the proof of Lemma \ref{graph}, so, we pass to the limit with all terms in Problem~\ref{P_trunc}. To finally remove the truncations, we need to check that
\begin{align*}
 h_\nu(N_n(u_n))\, &\to  h_\nu(u), \\[2mm]
 h_\tau(N_n(u_n),N_n(v_n),M_n(\theta_n))\, &\to h_\tau(u,v,\theta), \\[2mm]
h_w(N_n(u_n),N_n(v_n))\, & \to  h_w(u,v),
\end{align*}
strongly in $L^2(0,T;L^2(\Gamma_C))$. By continuity of $h_\nu, h_\tau,h_w$ it is enough to show that 
\begin{eqnarray}
&& N_n(v_n)\to v \quad \textrm{strongly in} \ L^2(0,T;L^2(\Gamma_C)), \label{trunc_1}\\[2mm]
&& N_n(u_n)\to u \quad \textrm{strongly in} \ L^2(0,T;L^2(\Gamma_C)), \label{trunc_2}\\[2mm]
&& M_n(\theta_n)\to \theta \quad \textrm{strongly in} \ L^2(0,T;L^2(\Gamma_C)).\label{trunc_3} 
\end{eqnarray}
Since by the Aubin-Lions lemma $v_n\to v$ strongly in $L^2(0,T;H^{1-\delta}(\Omega)^d)$, by continuity of the trace we have 
\begin{equation}
 v_n\to v \  \mbox{strongly in} \  {L^2((0,T)\times \Gamma_C)}. \label{convonboundary}
\end{equation}
From \eqref{convonboundary}, Lemma~\ref{Liptrunc} and the Lebesgue dominated convergence theorem, we obtain
\begin{eqnarray*}
	&& \|N_n(v_n)-v\|_{L^2((0,T)\times \Gamma_C)}^2 \le 2\|N_n(v_n)-N_n(v)\|_{L^2((0,T)\times \Gamma_C)}^2 + 2\|N_n(v)-v\|_{L^2((0,T)\times \Gamma_C)}^2  \\[2mm]
	&& \qquad \le 2\|v_n-v\|_{L^2((0,T)\times \Gamma_C)}^2  + 2\|N_n(v)-v\|_{L^2((0,T)\times \Gamma_C)}^2  \to 0 \ \textrm{as} \ n\to \infty.
\end{eqnarray*}
This proves \eqref{trunc_1}. To show \eqref{trunc_2} and \eqref{trunc_3}, we repeat similar calculations for $u_n$ and $\theta_n$. 
Hence, we can pass to the limit with truncation parameter $l=n\to \infty$ in all terms. This completes the proof of Theorem~\ref{mainth}.
\end{proof}

Thus, the model has at least one solution. The question of uniqueness remains unresolved, but in view of the
complexity of the system and its nonlinearities, it is unlikely. Indeed, the uniqueness of solution to Problem~\ref{P_V} 
does not follow from the presented argument, as it does in a case of the Banach fixed point theorem. Moreover, we suspect that uniqueness would require additional smallness assumptions on the data and stronger assumption on the functions $h_\tau,h_\nu,h_w$. 

As has been already mentioned, establishing an existence theorem for a purely elastic model is of considerable 
mathematical interest.

\end{document}